\documentclass[a4paper,UKenglish,cleveref, autoref, thm-restate]{lipics-v2021}
\usepackage{amsmath}
\usepackage{graphicx}

\usepackage{microtype}
\usepackage{booktabs} %
\usepackage{xcolor}         %
\usepackage{caption} %
\usepackage{csquotes}

\usepackage{algorithm}
\usepackage{algorithmic}
\usepackage{amssymb}
\usepackage{mathtools}
\usepackage{amsthm}
\usepackage{threeparttable}

\usepackage{tikz}
\graphicspath{ {./tikz/} }
\usetikzlibrary{positioning,graphs,matrix,backgrounds,decorations.pathreplacing,calc,patterns, arrows}
\tikzset{
    treevertex/.style = {align=center, scale=0.7, inner sep=1pt, text centered, circle, draw, minimum size=6mm, font={\small}},
  arn_x/.style = {treevertex, rectangle, draw=black,
  minimum width=0.5em, minimum height=0.5em}%
}

\newcommand{\N}{\mathbb{N}}

\newcommand{\R}{\mathbb{R}}
\newcommand{\M}{\mathcal{M}}
\newcommand{\D}{\mathcal{D}}
\newcommand{\T}{\mathcal{T}}

\newcommand{\nSigma}{\widehat{\Sigma}}

\newcommand{\norm}[1]{\left\lVert#1\right\rVert}

\newcommand{\lap}{\mathsf{Lap}}

\newcommand{\puredp}{$\epsilon$-DP}
\newcommand{\approxdp}{$(\epsilon,\delta)$-DP}
\DeclarePairedDelimiter\ceil{\lceil}{\rceil}
\DeclarePairedDelimiter\floor{\lfloor}{\rfloor}
\DeclarePairedDelimiter\abs{\lvert}{\rvert}
\DeclareMathOperator{\sgn}{sgn}

\DeclareMathOperator{\enc}{\mathsf{rep}}
\DeclareMathOperator{\nenc}{\mathsf{\widehat{rep}}}
\DeclareMathOperator{\eenc}{\mathsf{\widetilde{rep}}}
\DeclareMathOperator{\var}{Var}
\DeclareMathOperator{\mse}{MSE}

\DeclareMathOperator{\expectation}{\mathbb{E}}
\DeclareMathOperator{\tvd}{TV}
\newcommand{\Alg}{\mathrm{Alg}}

\renewcommand\vec{\mathbf}

\renewcommand\epsilon{\varepsilon}
\usepackage{makecell}

\title{Count on Your Elders: Laplace vs Gaussian Noise}

\author{Joel Daniel Andersson}{BARC, University of Copenhagen}{jda@di.ku.dk}{https://orcid.org/0000-0003-2530-0520}{}
\author{Rasmus Pagh}{BARC, University of Copenhagen}{pagh@di.ku.dk}{https://orcid.org/0000-0002-1516-9306}{}
\author{Teresa Anna Steiner}{University of Southern Denmark}{steiner@imada.sdu.dk}{https://orcid.org/0000-0003-1078-4075}{}
\author{Sahel Torkamani}{University of Edinburgh}{s.torkamani@sms.ed.ac.uk}{}{}

\Copyright{Joel Daniel Andersson, Rasmus Pagh, Teresa Anna Steiner, and Sahel Torkamani} %
\keywords{differential privacy, continual observation, streaming, prefix sums, trees}

\nolinenumbers
\hideLIPIcs

\ccsdesc{Security and privacy~Privacy-preserving protocols}

\authorrunning{J. D. Andersson, R. Pagh, T. A. Steiner, and S. Torkamani} %

\begin{document}
\maketitle

\begin{abstract}
In recent years, Gaussian noise has become a popular tool in differentially private algorithms, often replacing Laplace noise which dominated the early literature on differential privacy.
Gaussian noise is the standard approach to \emph{approximate} differential privacy, often resulting in much higher utility than traditional (pure) differential privacy mechanisms.
In this paper we argue that Laplace noise may in fact be preferable to Gaussian noise in many settings, in particular when we seek to achieve $(\varepsilon,\delta)$-differential privacy for small values of $\delta$.
We consider two scenarios:

First, we consider the problem of counting under continual observation and present a new generalization of the binary tree mechanism that uses a $k$-ary number system with \emph{negative digits} to improve the privacy-accuracy trade-off.
Our mechanism uses Laplace noise and whenever $\delta$ is sufficiently small it improves the mean squared error over the best possible
$(\varepsilon,\delta)$-differentially private factorization mechanisms based on Gaussian noise.
Specifically, using $k=19$ we get an asymptotic improvement over the bound given in the work by Henzinger, Upadhyay and Upadhyay (SODA 2023) when $\delta = O(T^{-0.92})$.

Second, we show that the noise added by the Gaussian mechanism can always be replaced by Laplace noise of comparable variance for the same $(\epsilon, \delta)$-differential privacy guarantee, and in fact for sufficiently small $\delta$ the variance of the Laplace noise becomes strictly better.
This challenges the conventional wisdom that Gaussian noise should be used for high-dimensional noise.

Finally, we study whether counting under continual observation may be easier in an average-case sense than in a worst-case sense. We show that, under pure differential privacy, the expected worst-case error for a random input must be $\Omega(\log(T)/\varepsilon)$, matching the known lower bound for worst-case inputs.
\end{abstract}

\section{Introduction}

In private data analysis the goal is to release the result of computations on datasets while safeguarding the privacy of the individuals whose data make up the dataset.
A popular framework for achieving this is \emph{differential privacy}~\cite{dp_2006}, which gives probabilistic guarantees on how much can be learned about an individual's data from the result of a computation.
Algorithms providing differential privacy (DP) guarantees conceptually introduce \emph{noise}, randomly perturbing the result of a corresponding exact, non-private computation.
The trade-off between utility and privacy guarantees depends on the nature of the computation.

Let $\vec{x}, \vec{x}'\in\{0, 1\}^T$ be input datasets.
We say $\vec{x}$ and $\vec{x}'$ are neighboring if they are equal except for a single bit.
We want to compute all prefix sums of $\vec{x}$ with differential privacy under this neighboring relation, i.e., we want to privately release $A\vec{x}$ where $A$ is the lower-triangular all 1s matrix.
If we add the restriction that the bits $\vec{x}_1, \vec{x}_2, \dots, \vec{x}_T$ are received one by one as a stream, and we are to output $(A\vec{x})_i$ on receiving $\vec{x}_i$, then this problem is referred to as \emph{counting under continual observation} or \emph{continual counting} for short \cite{dwork_differential_2010, chan_private_2011}.

The standard solution to this problem is the \emph{factorization mechanism} \cite{li_matrix_2015}, and it involves (1) factorizing $A$ into matrices $L, R$ such that $A = LR$, (2) privately releasing $R\vec{x}$ by adding noise $\vec{z}$ and then (3) multiplying the private release by $L$ from the left: $L(R\vec{x} + \vec{z}) = A\vec{x} + L\vec{z}$.
Defining the $\ell_p$-sensitivity as $\Delta_p = \max_{i\in[T]}\norm{R\vec{e}_i}_p$, where $\vec{e}_i$ is the $i$\textsuperscript{th} unit basis vector,
classical results on DP say that adding Laplace noise scaled to $\Delta_1$ yields $\epsilon$-DP (also called \emph{pure} DP), and adding Gaussian noise scaled to $\Delta_2$ yields $(\epsilon, \delta)$-DP (also called \emph{approximate} DP).
The scheme rests on the observation that while the sensitivity of $A$ may be large, the sensitivity of $R$ can be much smaller, and so adding noise $L\vec{z}$ to the final result may result in better accuracy than adding noise directly to $A\vec{x}$.

While continual counting is a deceptively simple problem, we note that it still has significant gaps between upper and lower bounds \cite{dwork_differential_2010, henzinger_almost_2023, fichtenberger_constant_2023, dwork_rectangle_queries_2015, cohen2024lower}. The problem is used as a subroutine in various applications ranging from private learning \cite{mcmahan_federated, kairouz_practical_2021, choquette-choo_multi-epoch_2022, denissov_improved_2022, choquettechoo2023amplified} to histogram estimation \cite{cardoso_differentially_2022, chan_differentially_2012, huang_frequency_2022, upadhyay_sublinear_2019}, so any improvement in upper bounds affect these applications.

\subsection{Our contributions}\label{sec:our_contributions}
Our first contribution is an easy-to-implement tree-aggregation-based mechanism for continual counting under $\epsilon$-DP, and based on Laplace noise, with properties stated in \cref{thm:main}.
\begin{theorem}\label{thm:main}
    Given a constant odd integer $k\geq 3$ and integer $T\geq 2$, there exists an $\epsilon$-DP algorithm for continual counting that for a stream $\vec{x}_1, \vec{x}_2, \dots, \vec{x}_T$ achieves mean squared error $\frac{k(1-1/k^2)}{2\epsilon^2\log(k)^3}\cdot\log(T)^3 + o(\log(T)^3)$, computing all outputs in time $O(T)$ and space $O(\log T)$.
\end{theorem}
We denote the base $k$ logarithm by $\log_k$, and define $\log(\cdot) = \log_2(\cdot)$.
To the best of our knowledge, this is the best bound on the mean squared error known for any $\epsilon$-DP mechanism, up to lower-order asymptotically vanishing terms.
It improves the error achieved by applying techniques in \cite{honaker2015} to the binary tree mechanism (as implemented in \cite{kairouz_practical_2021} and referred to elsewhere as \enquote{Honaker Online}) by a factor 4 when $k=19$.
While this is interesting in its own right, we note that this improvement in the leading constant has implications for how the error compares to \approxdp{} mechanisms.
Henzinger, Upadhyay and Upadhyay \cite{henzinger_almost_2023} produce a factorization with a bound on the mean squared error of $C_{\epsilon, \delta}^2\big(1 + \frac{\ln(4T/5)}{\pi}\big)^2$, where $C_{\epsilon, \delta} = \frac{2}{\epsilon}\sqrt{\frac{4}{9} + \ln\big(\frac{1}{\delta}\sqrt{\frac{2}{\pi}}\big)}$, and show that no factorization mechanism based on Gaussian noise can asymptotically improve on this error.\footnote{As the constant $C_{\epsilon, \delta}$ is valid for all $\epsilon, \delta \in (0, 1)$, it seems plausible that it can be reduced via a tighter analysis of the Gaussian mechanism, as done in for example \cite{pmlr-v80-balle18a}.}
We improve on their error using Laplace noise for small enough $\delta$.
\begin{theorem}\label{thm:beats_gaussian_noise}
    Given an integer $T\geq 2$, there exists an $\epsilon$-DP algorithm for continual counting, based on Laplace noise, that improves on the mean squared error bound given in \cite{henzinger_almost_2023}, up to lower-order asymptotically vanishing terms, whenever $\delta = O(T^{-0.92})$.
\end{theorem}
Given that $\delta$ is commonly set to be $o(1/T)$, this result is surprising.
The phenomenon is discussed in \Cref{sec:laplace_vs_gauss}, but in a nutshell a sufficiently small $\delta$ must cause the error to greatly increase if the mechanism is based on Gaussian noise.
While improving constants for \puredp{} is of independent interest, \Cref{thm:beats_gaussian_noise} implies that further improvements also have implications for \approxdp{}.

Motivated by the dominance of Gaussian noise in applications, and the moral in \Cref{thm:beats_gaussian_noise} for Laplace noise, our second contribution is $(\epsilon, \delta)$-DP guarantees for the Laplace mechanism for given $\ell_1$ and $\ell_2$ sensitivities.
\begin{theorem}[Approximate DP for the Laplace Mechanism]\label{theorem:laplace_both}
    Let $f : \mathcal{X}^n \to \R^d$ be a function with $\ell_p$-sensitivity $\Delta_p \coloneqq \max_{\D\sim\D'}\norm{f(\D)-f(\D')}_p$.
    For a given dataset $\D\in\mathcal{X}^n$, $\delta \in (0, 1)$ and $\lambda > \Delta_1$, the output $f(\D) + \lap(\lambda)^d$, satisfies $(\epsilon, \delta)$-differential privacy where
    \begin{equation*}
        \epsilon = \min\bigg\{ \frac{\Delta_1}{\lambda}, \frac{\Delta_2}{\lambda}\bigg(\frac{\Delta_2}{2\lambda} + \sqrt{2\ln(1/\delta)}\bigg)\bigg\}\,.
    \end{equation*}
\end{theorem}
Essentially we can implement \approxdp{} for the Laplace mechanism where the resulting $\epsilon$ is never worse than what simple composition gives, while also being competitive with the Gaussian mechanism in the regime where it is advantageous.

Finally, in \Cref{sec:lowerbound} we show the following lower bound, which matches the classical packing lower bound for continual observation under pure differential privacy~\cite{dwork_differential_2010, dwork_algorithmic_2013} in the case where the input is random.
\begin{theorem}\label{thm:lower}
    Let $\mathcal{U}$ be the distribution of strings of length $T$ where every element in the string is drawn from $B(1/2)$ (Bernoulli distribution with $p=1/2$).
    Let $\mathcal{M}$ be an $\varepsilon$-differentially private algorithm that solves the continual counting problem for strings from $\mathcal{U}$ with error at most $\alpha$ with probability at least $3/4$ (probability over random input and random choices of the algorithm).
    Then $\alpha=\Omega(\varepsilon^{-1}\log T)$.
\end{theorem}

\subsection{Overview of technical ideas}\label{sec:tech_sketch}

\paragraph*{Improved tree-based mechanism}
We start out by building a complete $k$-ary tree of height $h=\ceil{\log_k(T+1)}$ on top of the input $\vec{x}\in\{0, 1\}^T$, see \Cref{fig: b-ary}.
To analyze this, it turns out that expressing a time step $t$ as a $k$-ary number, i.e., as $h$ digits in $\{0, \dots, k-1\}$, i.e.\ $\vec{t}\in\{0, \dots, k-1\}^h$, gives a direct means of computing which vertices contribute to a given prefix sum.
Increasing the arity of the tree reduces the height, which reduces the $\ell_1$-sensitivity and levels of vertices that might get added, but increases the number of vertices to add per level.
To mitigate the increase in number of vertices to add, we leverage \emph{subtraction}, see \Cref{fig: ternary}.

To analyze the structure of the resulting prefix sums when subtraction is introduced, it turns out that rather than expressing $t$ in digits from $0$ to $k-1$, we should instead allow for \emph{negative digits}.
Focusing on the case of odd $k$, we thus represent $t$ as $\vec{t}\in\{-\frac{k-1}{2},\dots, \frac{k-1}{2}\}^h$.
The interpretation of $\vec{t}$ becomes: if $\vec{t}_\ell \geq 0$, add $\vec{t}_\ell$ leftmost children on level $\ell$, otherwise, \emph{subtract} $|\vec{t}_\ell|$ rightmost children on level $\ell$.
Identifying that $\norm{\vec{t}}_1$ is equal to the number of vertices used to produce the estimate at time $t$, we can analyze the average number of vertices added for an estimate from this representation by studying $\vec{t}$ for $t\in[1, T]$.
This gives, roughly, a factor-2 reduction in the mean squared error versus not using subtraction.

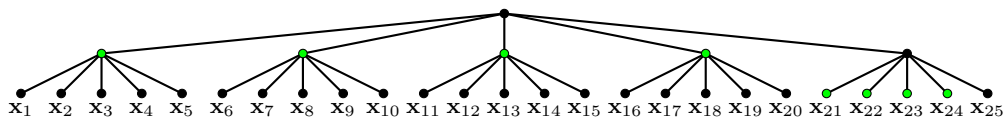
\begin{figure*}[b]
    \centering
\begin{tikzpicture}[scale=0.53]

\draw[color=black, thick] (1,1) -- (3,2);
\draw[color=black, thick] (2,1) -- (3,2);
\draw[color=black, thick] (3,1) -- (3,2);
\draw[color=black, thick] (4,1) -- (3,2);
\draw[color=black, thick] (5,1) -- (3,2);

\draw[color=black, thick] (6,1) -- (8,2);
\draw[color=black, thick] (7,1) -- (8,2);
\draw[color=black, thick] (8,1) -- (8,2);
\draw[color=black, thick] (9,1) -- (8,2);
\draw[color=black, thick] (10,1) -- (8,2);

\draw[color=black, thick] (11,1) -- (13,2);
\draw[color=black, thick] (12,1) -- (13,2);
\draw[color=black, thick] (13,1) -- (13,2);
\draw[color=black, thick] (14,1) -- (13,2);
\draw[color=black, thick] (15,1) -- (13,2);

\draw[color=black, thick] (16,1) -- (18,2);
\draw[color=black, thick] (17,1) -- (18,2);
\draw[color=black, thick] (18,1) -- (18,2);
\draw[color=black, thick] (19,1) -- (18,2);
\draw[color=black, thick] (20,1) -- (18,2);

\draw[color=black, thick] (21,1) -- (23,2);
\draw[color=black, thick] (22,1) -- (23,2);
\draw[color=black, thick] (23,1) -- (23,2);
\draw[color=black, thick] (24,1) -- (23,2);
\draw[color=black, thick] (25,1) -- (23,2);

\draw[color=black, thick] (3,2) -- (13,3);
\draw[color=black, thick] (8,2) -- (13,3);
\draw[color=black, thick] (13,2) -- (13,3);
\draw[color=black, thick] (18,2) -- (13,3);
\draw[color=black, thick] (23,2) -- (13,3);

\draw[black,fill=black] (1,1) circle (.7ex);
\node at (1,0.6) {{\footnotesize $\vec{x}_1$}};
\draw[black,fill=black] (2,1) circle (.7ex);
\node at (2,0.6) {{\footnotesize $\vec{x}_2$}};
\draw[black,fill=black] (3,1) circle (.7ex);
\node at (3,0.6) {{\footnotesize $\vec{x}_3$}};
\draw[black,fill=black] (4,1) circle (.7ex);
\node at (4,0.6) {{\footnotesize $\vec{x}_4$}};
\draw[black,fill=black] (5,1) circle (.7ex);
\node at (5,0.6) {{\footnotesize $\vec{x}_5$}};
\draw[black,fill=black] (6,1) circle (.7ex);
\node at (6,0.6) {{\footnotesize $\vec{x}_6$}};
\draw[black,fill=black] (7,1) circle (.7ex);
\node at (7,0.6) {{\footnotesize $\vec{x}_7$}};
\draw[black,fill=black] (8,1) circle (.7ex);
\node at (8,0.6) {{\footnotesize $\vec{x}_8$}};
\draw[black,fill=black] (9,1) circle (.7ex);
\node at (9,0.6) {{\footnotesize $\vec{x}_9$}};
\draw[black,fill=black] (10,1) circle (.7ex);
\node at (10,0.6) {{\footnotesize $\vec{x}_{10}$}};
\draw[black,fill=black] (11,1) circle (.7ex);
\node at (11,0.6) {{\footnotesize $\vec{x}_{11}$}};
\draw[black,fill=black] (12,1) circle (.7ex);
\node at (12,0.6) {{\footnotesize $\vec{x}_{12}$}};
\draw[black,fill=black] (13,1) circle (.7ex);
\node at (13,0.6) {{\footnotesize $\vec{x}_{13}$}};
\draw[black,fill=black] (14,1) circle (.7ex);
\node at (14,0.6) {{\footnotesize $\vec{x}_{14}$}};
\draw[black,fill=black] (15,1) circle (.7ex);
\node at (15,0.6) {{\footnotesize $\vec{x}_{15}$}};
\draw[black,fill=black] (16,1) circle (.7ex);
\node at (16,0.6) {{\footnotesize $\vec{x}_{16}$}};
\draw[black,fill=black] (17,1) circle (.7ex);
\node at (17,0.6) {{\footnotesize $\vec{x}_{17}$}};
\draw[black,fill=black] (18,1) circle (.7ex);
\node at (18,0.6) {{\footnotesize $\vec{x}_{18}$}};
\draw[black,fill=black] (19,1) circle (.7ex);
\node at (19,0.6) {{\footnotesize $\vec{x}_{19}$}};
\draw[black,fill=black] (20,1) circle (.7ex);
\node at (20,0.6) {{\footnotesize $\vec{x}_{20}$}};
\draw[black,fill=green] (21,1) circle (.7ex);
\node at (21,0.6) {{\footnotesize $\vec{x}_{21}$}};
\draw[black,fill=green] (22,1) circle (.7ex);
\node at (22,0.6) {{\footnotesize $\vec{x}_{22}$}};
\draw[black,fill=green] (23,1) circle (.7ex);
\node at (23,0.6) {{\footnotesize $\vec{x}_{23}$}};
\draw[black,fill=green] (24,1) circle (.7ex);
\node at (24,0.6) {{\footnotesize $\vec{x}_{24}$}};
\draw[black,fill=black] (25,1) circle (.7ex);
\node at (25,0.6) {{\footnotesize $\vec{x}_{25}$}};

\draw[black,fill=green] (3,2) circle (.7ex);
\draw[black,fill=green] (8,2) circle (.7ex);
\draw[black,fill=green] (13,2) circle (.7ex);
\draw[black,fill=green] (18,2) circle (.7ex);
\draw[black,fill=black] (23,2) circle (.7ex);
\draw[black,fill=black] (13,3) circle (.7ex);

\end{tikzpicture}     \label{fig: 5-ary}
\caption{Complete 5-ary tree with 25 leaves containing inputs $\vec{x}_1,\dots,\vec{x}_{25}$ and inner vertices containing subtree sums.
To compute the sum of the first 24 inputs we can add 4 inner vertices and 4 leaves (shown in green) --- in general, the worst-case number of vertices for $k$-ary trees is $k-1$ per level.
When subtree sums are made private by random (Laplace) noise the variance of a prefix sum estimator is proportional to the number of terms added.
In \cref{sec:kary_trees} we analyze the privacy and utility of this natural generalization of the binary tree mechanism to $k$-ary trees, among other things showing that the mean number of terms is about half of the worst case.
}
\label{fig: b-ary}
\end{figure*}
\begin{figure*}[h!b]
\begin{subfigure}[h]{0.45\textwidth}
    \centering
\begin{tikzpicture}[scale=0.65]

\draw[color=black, thick] (1,1) -- (2,2);
\draw[color=black, thick] (2,1) -- (2,2);
\draw[color=black, thick] (3,1) -- (2,2);

\draw[color=black, thick] (4,1) -- (5,2);
\draw[color=black, thick] (5,1) -- (5,2);
\draw[color=black, thick] (6,1) -- (5,2);

\draw[color=black, thick] (7,1) -- (8,2);
\draw[color=black, thick] (8,1) -- (8,2);
\draw[color=black, thick] (9,1) -- (8,2);

\draw[color=black, thick] (2,2) -- (5,3);
\draw[color=black, thick] (5,2) -- (5,3);
\draw[color=black, thick] (8,2) -- (5,3);

\draw[black,fill=black] (1,1) circle (.7ex);
\node at (1,0.6) {{\footnotesize $\vec{x}_1$}};
\draw[black,fill=black] (2,1) circle (.7ex);
\node at (2,0.6) {{\footnotesize $\vec{x}_2$}};
\draw[black,fill=black] (3,1) circle (.7ex);
\node at (3,0.6) {{\footnotesize $\vec{x}_3$}};
\draw[black,fill=black] (4,1) circle (.7ex);
\node at (4,0.6) {{\footnotesize $\vec{x}_4$}};
\draw[black,fill=black] (5,1) circle (.7ex);
\node at (5,0.6) {{\footnotesize $\vec{x}_5$}};
\draw[black,fill=black] (6,1) circle (.7ex);
\node at (6,0.6) {{\footnotesize $\vec{x}_6$}};
\draw[black,fill=green] (7,1) circle (.7ex);
\node at (7,0.6) {{\footnotesize $\vec{x}_7$}};
\draw[black,fill=green] (8,1) circle (.7ex);
\node at (8,0.6) {{\footnotesize $\vec{x}_8$}};
\draw[black,fill=black] (9,1) circle (.7ex);
\node at (9,0.6) {{\footnotesize $\vec{x}_9$}};

\draw[black,fill=green] (2,2) circle (.7ex);
\draw[black,fill=green] (5,2) circle (.7ex);
\draw[black,fill=black] (8,2) circle (.7ex);
\draw[black,fill=black] (5,3) circle (.7ex);

\end{tikzpicture}     \caption{Prefix sum with only vertex addition}
    \label{fig: ter2}
\end{subfigure}
\hfill
\begin{subfigure}[h]{0.45\textwidth}
    \centering
\begin{tikzpicture}[scale=0.65]

\draw[color=black, thick] (1,1) -- (2,2);
\draw[color=black, thick] (2,1) -- (2,2);
\draw[color=black, thick] (3,1) -- (2,2);

\draw[color=black, thick] (4,1) -- (5,2);
\draw[color=black, thick] (5,1) -- (5,2);
\draw[color=black, thick] (6,1) -- (5,2);

\draw[color=black, thick] (7,1) -- (8,2);
\draw[color=black, thick] (8,1) -- (8,2);
\draw[color=black, thick] (9,1) -- (8,2);

\draw[color=black, thick] (2,2) -- (5,3);
\draw[color=black, thick] (5,2) -- (5,3);
\draw[color=black, thick] (8,2) -- (5,3);

\draw[black,fill=black] (1,1) circle (.7ex);
\node at (1,0.6) {{\footnotesize $\vec{x}_1$}};
\draw[black,fill=black] (2,1) circle (.7ex);
\node at (2,0.6) {{\footnotesize $\vec{x}_2$}};
\draw[black,fill=black] (3,1) circle (.7ex);
\node at (3,0.6) {{\footnotesize $\vec{x}_3$}};
\draw[black,fill=black] (4,1) circle (.7ex);
\node at (4,0.6) {{\footnotesize $\vec{x}_4$}};
\draw[black,fill=black] (5,1) circle (.7ex);
\node at (5,0.6) {{\footnotesize $\vec{x}_5$}};
\draw[black,fill=black] (6,1) circle (.7ex);
\node at (6,0.6) {{\footnotesize $\vec{x}_6$}};
\draw[black,fill=black] (7,1) circle (.7ex);
\node at (7,0.6) {{\footnotesize $\vec{x}_7$}};
\draw[black,fill=black] (8,1) circle (.7ex);
\node at (8,0.6) {{\footnotesize $\vec{x}_8$}};
\draw[black,fill=red] (9,1) circle (.7ex);
\node at (9,0.6) {{\footnotesize $\vec{x}_9$}};

\draw[black,fill=black] (2,2) circle (.7ex);
\draw[black,fill=black] (5,2) circle (.7ex);
\draw[black,fill=black] (8,2) circle (.7ex);
\draw[black,fill=green] (5,3) circle (.7ex);

\end{tikzpicture}     \caption{Prefix sum with vertex subtraction}
    \label{fig: ter1}
\end{subfigure}
\caption{Complete ternary trees, each with 9 leaves containing inputs $x_1,\dots,x_{9}$ and inner vertices containing subtree sums.
The trees illustrate two different possibilities for computing the sum of the first 8 inputs:
(\ref{fig: ter2}) add two inner vertices and two leaf vertices, or (\ref{fig: ter1}) subtract one vertex (shown in red) from the sum stored in the root vertex (shown in green).
The variance of using the latter method, with subtraction, is half of the former method, assuming the noise distribution for each vertex is the same.
In \cref{sec:ksub} we analyze the privacy and utility of our new generalization of the binary tree mechanism to $k$-ary trees that makes use of subtraction.
Among other things we show that the worst-case number of terms needed per level is $\lfloor k/2\rfloor$, and the mean number of terms is about half of the worst case.}
\label{fig: ternary}
\end{figure*}
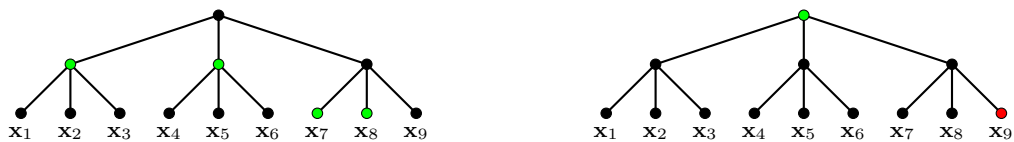

\paragraph*{Factorization mechanism with Laplace noise}
To derive the \approxdp{} guarantee of the Laplace mechanism (Theorem~\ref{theorem:laplace_both}), we make use of a heterogeneous composition theorem \cite{kairouz15}.
Given two arbitrary multidimensional neighboring inputs, we argue for the privacy loss in each coordinate when one of them is released with respect to its neighbor.
To get the aggregate privacy guarantee for the entire vector, we apply the composition theorem over the individual coordinates being released, and show that the guarantee must hold for any two neighboring outputs.
It turns out that, as in the case of the Gaussian mechanism, the \approxdp{} guarantee for Laplace noise naturally depends on the $\ell_2$-sensitivity of the function being released.

\paragraph*{Lower bound for random inputs}
Finally, the lower bound of Theorem~\ref{thm:lower} is shown by a variation of the packing argument that has been used to show an $\Omega(\log(T)/\epsilon)$ lower bound for worst-case inputs. 
The core idea is to construct a set of \emph{correlated} input strings that are individually close to uniformly random, but a continual observation mechanism that has low maximum error on all of them would be able to reliably distinguish these input strings.
More precisely, for integers $T$, $B$, $m=T/B$ and $k\leq B/4$, we construct our collection of strings $\vec{x}^{(0)}, \vec{x}^{(1)}, \dots, \vec{x}^{(m)}$ by first decomposing $[T]$ into blocks of size $B : B_1 = [1, B], B_2 = [B+1, 2B], \dots, B_m = [T-B+1, T]$ and then doing the following sampling process:
$\vec{x}^{(0)}$ is sampled uniformly from $\{0, 1\}^T$ conditioned on each block $B_i$ containing between $B/4$ and $3B/4$ 1s,
and for $i\in[m]$, $\vec{x}^{(i)}$ is derived from $\vec{x}^{(0)}$ by flipping $k$ 0s to 1s in block $B_i$.
We show that each individual string has $o(1)$ total variation distance to the uniform distribution over $\{0, 1\}^T$, and thus a mechanism accurate on the uniform distribution will also be accurate on our random strings. Also, pairs of our random strings will be $2k$-neighboring with probability 1.
The key insight is that the difference of the true count of $\vec{x}^{(i)}$ minus $\vec{x}^{(0)}$ is $k$ at the end of $B_i$, and zero at the end of each preceding $B_j$ for $j < i$, and this event is detectable if the mechanism is sufficiently accurate on both inputs.
Choosing parameters carefully, a packing-like argument in the end shows that a too low error on uniformly distributed inputs contradicts differential privacy.

\subsection{Related work}\label{sec:related_work}

Throughout this section, we focus on \emph{unbiased} mechanisms for continual counting with \enquote{good error}.
Here \enquote{error} is a shorthand for mean squared error over~$T$ outputs, and \enquote{good} is an error of $O(\log(T)^3)$ for $\epsilon$-DP mechanisms and $O(\log(T)^2\log(1/\delta))$ for $(\epsilon, \delta)$-DP mechanisms.
We note that there exists some biased mechanisms \cite{dwork_rectangle_queries_2015}, but they focus on reducing the error for \emph{sparse} inputs and give no improvement for dense streams.
Unless stated otherwise, constant factor improvements are always for $\epsilon$-DP.

Historically, the first mechanisms with good error are typically credited to Chan, Shi, and Song~\cite{chan_private_2011} and Dwork, Naor, Pitassi, and Rothblum~\cite{dwork_differential_2010}, though similar constructions were proposed independently by others \cite{dp_histograms_2010, dp_wavelet_2010}.
These results are commonly referred to as \emph{binary tree mechanisms}. 
Honaker~\cite{honaker2015}
observed that the full tree used by the binary tree mechanism can be used to produce more efficient estimates since each prefix query can be computed using different combinations of vertices.
A subset of Honaker's techniques can be used to reduce the mean squared error by up to a factor 2, at some cost in efficiency. %
Rather than leveraging the density of the tree to improve the error, Andersson and Pagh~\cite{andersson2024smooth} identified that making the tree sparser -- only using a subset of the leaves to store inputs -- allowed for a factor 2 reduction in the error while keeping the variance of outputs constant, for $(\epsilon, \delta)$-DP.
Although they do not state a result for it, using their technique for $\epsilon$-DP yields a factor 4 improvement over the binary tree mechanism, which is still more than a factor 2 away from our improvement.
We also mention the work of Qardaji, Yang, and Li~\cite{qardaji_2013} which, while focusing on range queries, did investigate reducing the mean squared error by means of considering trees of higher arity.
They claim an optimal arity of $k=16$, but average over a different set of queries, and do not leverage the subtraction of vertices.
Cardoso and Rogers~\cite{cardoso_differentially_2022} used $k$-ary trees for prefix sums as a subroutine for \approxdp{} and investigated choosing $k$ to minimize the maximum variance.

Many recent papers \cite{dj_near_optimal_noise_generation_24, denissov_improved_2022, henzinger_almost_2023, fichtenberger_constant_2023} have focused on improving the error for continual counting for $(\epsilon, \delta)$-DP by directly searching among factorizations of $A$, the lower-triangular all-1s matrix.
Most notably, Henzinger, Upadhyay, and Upadhyay~\cite{henzinger_almost_2023} showed that there exists an optimal factorization $L=R=\sqrt{A}$ (also identified by Bennett~\cite{bennett77}) whose error is \emph{optimal up to the leading constant} across all factorizations. But, importantly, this is for $(\epsilon, \delta)$-DP and assumes using Gaussian noise.
Due to the high $\ell_1$-sensitivity of $R=\sqrt{A}$, the noise needed to release $R\vec{x}$ with the (conventional) Laplace mechanism results in an \puredp{} matrix mechanism with worse error than that of the binary mechanism.

The second part of our paper explores the use of Laplace noise for \approxdp{}, especially in the small $\delta$ regime -- a regime we are not the first to investigate \cite{SteinkeU16}.
We claim no novelty in pointing out that the Gaussian mechanism scales poorly for small $\delta$ (see e.g.\ \cite{DongSZ21}).
We do note that there is a rich literature investigating alternative (non-Gaussian) noise distribution for approximate DP \cite{GengDGK20, HolohanABA20, Liu19, dagan22a, GengDGK19,Vinterbo22}.
In \cite{SommerMM19} they compare how Laplace and Gaussian noise behaves over many compositions, and especially observe that eventually they will both converge to Gaussian privacy loss distributions.

Finally, we turn to lower bounds where the lower bound of~\cite{dwork_differential_2010} remains the best bound known on the noise needed for continual observation under pure differential privacy.
The hard inputs used in the proof (see~\cite{dwork_algorithmic_2013} for details) are very sparse vectors with a single block of 1s that can be reliably distinguished by any algorithm that outputs prefix sums with noise significantly smaller than the block size, yet are close to each other in neighbor distance.
This structure means that a ``packing argument'' can be used to show a lower bound on the amount of noise needed.
One might wonder if \emph{most} inputs are easier and allow a smaller error --- our lower bound shows that this is not the case.
Some algorithmic problems have the property that worst-case problem instances can be reduced to random instances (for example~\cite{henzinger2022_subgraphs}).
Though we are not aware of such a reduction for continual counting, our lower bound can be seen as an instantiation of this general principle.

\section{Preliminaries}\label{sec:prel}

Before introducing any mechanism, we first introduce some notation and background.
We use $\N$ to denote all positive integers, and for $a, b\in\N$ let $[a, b]\subset\N$ refer to the set of integers $\{a, a+1\dots, b\}$.
A vector-valued quantity $\mathbf{y}\in\mathbb{R}^n$ is expressed in bold notation and its scalar entries are written as $\mathbf{y}_i$ for $i\in [1, n]$.
Standard definitions for \puredp{} and \approxdp{} are given in \cref{appendix:dp}.

As the main goal of this paper is the study and design of differentially private algorithms for counting queries, we will first formalize the problem.
Let vector $\vec{x}\in \{0, 1\}^T$ be an input stream, where $\vec{x}_t$ is received at time $t$. Then, two input streams $\vec{x}, \vec{x}'$ are \emph{neighboring inputs}, denoted $\vec{x}\sim\vec{x}'$, if they take on the same value at every time step except one.
The target function we wish to privately release is $f : \{0, 1\}^T \to \R^T$ where $f(\vec{x})_t = \sum_{i=1}^t \vec{x}_i$.
We say that a random algorithm $\M: \{0, 1\}^T \to \mathbb{R}^T$ is a differentially private mechanism for continual counting if it satisfies either \puredp{} or \approxdp{} under this neighboring relation.

For the entirety of this paper, we will express the utility of an algorithm in terms of its \emph{mean squared error} defined as:
$\mse(\M, \vec{x}) = \frac{1}{T}\expectation\big[\norm{\M(x) - f(x)}_2^2\big] = \frac{1}{T}\sum_{t=1}^T \var [\M(\vec{x})_t]\,,$
where the second equality only holds if $\forall \vec{x}\in\{0, 1\}^T : \expectation[\M(\vec{x})] = f(\vec{x})$, which is the case for all factorization mechanisms in general, and all mechanisms covered in this work in particular.
That is to say, for unbiased mechanisms $\M$, their mean squared error will be equal to the average variance.
Moving forward, any mention of error without further qualification will refer to the mean squared error.

\section{Mechanisms}
In this section, we will gradually build up towards the algorithm outlined in \cref{thm:main}.
We begin by describing the binary tree mechanism, which will be our starting point.
Then we will consider how the utility improves when we extend to $k$-ary trees, and we give an algorithm and results for this case.
Finally, our main contribution will be to incorporate \emph{subtraction} of vertices in $k$-ary trees to further improve the utility.
The formal proofs leading to \cref{thm:main} are given in the next section.

\subsection{Binary Tree Mechanisms}\label{sec:binary_mech}
The first algorithms~\cite{chan_private_2011,dwork_differential_2010, dp_wavelet_2010, dp_histograms_2010} with good utility are based on binary trees and are referred to as the \emph{binary tree mechanism}. %
Implementation details vary, but we will consider a formulation of it where only the left subtrees are used, and in particular not the root vertex, structurally most similar to Algorithm~2 in \cite{chan_private_2011}.
The main idea for the algorithm is to build a binary tree of height $h=\ceil{\log(T+1)}$ on top of the input stream $\vec{x}$ where each leaf stores an input $\vec{x}_t$, and each vertex higher up in the tree stores the sum of its children.
Having such a tree, it is possible to add together vertices in the tree in such a way that at most $1$ vertex per level in the tree is used, up to $h$ in total, where the sum corresponds to a prefix sum up to $t\in[1, T]$.
In particular, which vertices that get summed to produce a given prefix sum is described exactly by the $h$-bit \emph{binary representation} of the current time step, where each $1$ corresponds to adding a vertex in the tree to the output.
To make this private, we add i.i.d.\ Laplace noise $\lap(\Delta_1/\epsilon)$ to each vertex in the tree, where $\Delta_1=h$ since the trees for two neighboring inputs will differ in one vertex per level, excluding the root which is not used in computations.
We defer making exact statements about the corresponding utility to \cref{sec:kary_trees}, where this description of the binary tree mechanisms is a special case of \cref{alg:nosub} for $k=2$.

\subsection{\texorpdfstring{$k$}{k}-ary Tree Mechanisms}\label{sec:kary_trees}
We will next consider trees of greater arity $k\geq 2$ to improve constants for the error.
We note that using trees of higher arity has already been studied~\cite{qardaji_2013, cormode_range_queries_2019, cardoso_differentially_2022}.
With a greater arity tree, a lesser height is needed to accommodate a given number of leaves.
However, it might also require adding more vertices per level to generate a given prefix sum estimate.
To explore this trade-off, we introduce notation for $k$-ary trees next.
\begin{definition}[$k$-ary trees]\label{def:trees}
    For integers $k\geq 2$ and $h\geq 1$, we define a $k$-ary tree of height $h$ with $k^h$ leaves as $\T_{k,h}$. 
    Then for integer $1\leq\ell\leq h+1$, let $\T_{k,h}^\ell = \{ [1 + j\cdot k^{\ell-1}, (j+1)\cdot k^{\ell-1}] : j\in\mathbb{N},\, 0 \leq j < k^{h-\ell+1}\}$ be the set of vertices on level $\ell$, where a vertex $I\in\T_{k, h}^\ell$ is a child of $I'\in\T_{k, h}^{\ell + 1}$ if $I\subset I'$, and where $\T_{k, h} = \bigcup_{1 \leq\ell\leq h+1}\T^\ell_{k, h}$.
\end{definition}
Additionally, let $\Sigma(I) = \sum_{i\in I} \vec{x}_i$ for $I\in\T_k$ when such a tree is used to store partial sums on $\vec{x}$, $\nSigma$ for the case where a vertex in the tree stores a partial sum plus noise.
Also, we will find it expedient to express integers using $k$-ary digits, which we define next.
\begin{definition}[$k$-ary representation of integer]
    Let $k$, $w$ and $t \leq k^w - 1$ be positive integers.
    We then define $\vec{t} = \enc_k(t, w)\in[0, k-1]^w$ as the unique vector in $[0, k-1]^w$ satisfying $\sum_{i=1}^w k^{i-1}\vec{t}_i = t$.
\end{definition}
The intuition behind using $k$-ary to analyse $k$-ary tree mechanisms is the same as in the case where $k=2$ -- there is a clean analogue between the two.
The $h$-digit $k$-ary representation of a time step $t$, $\vec{t}=\enc_k(t, h)$, encodes a set of vertices in $\T_{k, h}$ whose sum (when interpreted as partial sums) equals the prefix sum up to $t$.
In particular, $\vec{t}_\ell$ gives the number of vertices that get added at level $\ell$ of the tree.
With this observation in mind, we introduce \cref{alg:nosub}.
\begin{algorithm}[htb]
   \caption{$k$-ary Tree Mechanism}\label{alg:nosub}

   \begin{algorithmic}[1] %
    \STATE {\bfseries Input:} $\vec{x}\in \{0, 1\}^T$, $k\geq 2$, privacy parameter $\epsilon$
    \STATE $h\gets \ceil{\log_k(T+1)}$
    \STATE Construct $k$-ary tree $\T_{k,h}$.
    \STATE For all $I\in\T_{k,h}$ compute $\nSigma(I) = \Sigma(I) + z$ for i.i.d.\ $z\sim\lap(h/\epsilon)$.
    \FOR{$t=1$ {\bfseries to} $T$}
        \STATE $count \gets 0$
        \STATE $\vec{t} \gets \enc_k(t, h)$
        \STATE $p \gets 0$
        \FOR{$\ell=h$ {\bfseries to} $1$}
            \FOR{$j=1$ {\bfseries to} $\vec{t}_\ell$}
                \STATE $count \gets count + \widehat{\Sigma}([p + 1, p + k^{\ell-1}])$
                \STATE $p \gets p + k^{\ell-1}$
            \ENDFOR
        \ENDFOR
        \STATE At time $t$, output $count$
    \ENDFOR
    \end{algorithmic}
\end{algorithm}

Now we state the properties of \cref{alg:nosub}.
We begin by discussing the privacy of the output of the algorithm.
Subsequently, we examine the variance and average number of used vertices to analyze the mean squared error.
Finally, we analyze the time and space complexity of the algorithm.
\begin{lemma}\label{lemma:nosub_private}
    The output from \cref{alg:nosub} satisfies $\epsilon$-DP.
\end{lemma}
\begin{proof}
    Note that line~11 always adds the noisy count from a vertex in the tree at level $\ell\in [1, h]$  since $p=p_n\overset{\pmod {k^{\ell-1}}}{=} 0$, and since $\vec{t}$ is a valid representation of a time step $t \leq T$, we also have $p\leq t$, and so the vertex must be contained in the tree.
    Note further that everything after line~4 can be considered post-processing of privately releasing the tree $\T_{k, h}$ storing partial sums.
    Since we never use the root vertex of the tree at level $h + 1$, any two neighboring inputs will differ at exactly $h$ vertices which we use, therefore $\Delta_1 = h$, implying that releasing the individual counts in the tree with $\lap(h/\epsilon)$ gives privacy.
\end{proof}
For the utility, note that the output from \cref{alg:nosub} at time $t$ has variance $(2h^2/\epsilon^2)\cdot\norm{\enc_k(t, h)}_1$, since $\norm{\enc_k(t, h)}_1$ vertices with variance $(2h^2/\epsilon^2)$ are added up. Consequently, we can conclude this subsection with the following Lemma and Theorem.

\begin{lemma}\label{lemma:nosub_mse}
    \cref{alg:nosub} achieves a mean squared error of $\frac{(k-1)h^3}{\epsilon^2(1 - 1/k^h)}$ over $T=k^h - 1$ outputs.
\end{lemma}
\begin{proof}
    Identifying that the mean squared error over $T$ steps is equal to the average number of vertices added together to produce an output, call this number $a_{[1, T]}$, times the variance of each vertex, $2h^2/\epsilon^2$.
    Observe that $a_{[1, T]} = \frac{1}{T}\sum_{t=1}^T \norm{\enc_k(t, h)}_1$ and that $T$ is the greatest integer possible to represent with $h$ digits from $\{0, 1, \dots, k-1\}$.
    This implies that we can compute the average $\ell_1$-weight over $\{0, 1,\dots, k-1\}^h$, $a_{[0, T]} = h\cdot\frac{1}{k}\sum_{d=0}^{k-1} |d| = \frac{h(k-1)}{2}$.
    Since $T\cdot a_{[1, T]} = (T + 1)\cdot a_{[0, T]}$, we have that $a_{[1, T]} = (1 + 1/T)\cdot a_{[0, T]} = \frac{h(k-1)}{2(1 - 1/k^h)}$.
    Plugging this in gives a mean squared error of $\frac{(k-1)h}{2(1-1/k^h)}\cdot \frac{2h^2}{\epsilon^2} = \frac{(k-1)h^3}{\epsilon^2(1 - 1/k^h)}$.
\end{proof}
\begin{theorem}\label{thm:no_subtraction}
    Given integer $T\geq 2$ and constant integer $k\geq 2$, there exists an $\epsilon$-DP mechanism for continual counting that on receiving a stream $\vec{x}_1, \vec{x}_2, \dots, \vec{x}_T$ achieves a mean squared error of $\frac{k-1}{\epsilon^2\log(k)^3}\cdot\log(T)^3 + o(\log(T)^3)$, computes all outputs in time $O(T)$ and $O(\log T)$ space.
\end{theorem}
As \cref{thm:no_subtraction} is not our main contribution, we provide a proof sketch in \cref{appendix:proofT1}.
Before continuing and including subtraction as a technique, we note that \cref{thm:no_subtraction} describes the standard binary tree mechanism for $k=2$, and that the leading constant in the error is minimized for $k=17$, in which case it equals $\approx 0.234$.

\subsection{\texorpdfstring{$k$}{k}-ary trees with subtraction}\label{sec:ksub}

Next, we will investigate what further improvements are possible if we allow for the subtraction of vertices to produce outputs.
The reason for doing so is rather simple: if you need to add $\geq\ceil{k/2}$ left-children for an estimate, you might as well add the parent and subtract the $\leq\floor{k/2}$ right-most children.
Figure \ref{fig: ternary} provides a visual example that further illustrates how incorporating vertex subtraction can effectively reduce the number of vertices required to compute the prefix sum.
Intuitively this should buy us another factor 2 improvement, as instead of adding a maximum of $(k-1)$ vertices per level to produce an estimate, we should now instead combine at most $\ceil{(k-1)/2}$.

To analyse (and efficiently implement) such a mechanism, the previous integer representation using digits in $[0, k-1]$ is no longer the natural choice.
Instead, when we open up the possibility of \emph{subtracting vertices}, the solution becomes to allow for \emph{negative digits}.
\begin{definition}[Offset $k$-ary representation of integers]\label{def:basis_offset}
    Let $k$, $w$ and $t$ be positive integers, with $k$ also being odd, satisfying $t \leq (k^w-1)/2$.
    We then define $\vec{t} = \nenc_k(t, w)\in[-\frac{k-1}{2}, \frac{k-1}{2}]^w$ as the unique vector with integer entries in $[-\frac{k-1}{2}, \frac{k-1}{2}]^w$ satisfying $\sum_{i=1}^{w}k^{i-1}\vec{t}_i = t$.
\end{definition}
We will return to the fact that \cref{def:basis_offset} only considers odd $k$ and first focus on providing intuition for why we are using offset digits.
As in the case for $k$-ary trees using only subtraction, there is a natural connection between $\vec{t}=\nenc_k(t, h)$ and vertices in $\T_{k, h}$.
Put simply: if $\vec{t}_\ell > 0$, then to produce the prefix sum for $t$, we will be adding $\vec{t}_\ell \leq (k-1)/2$ \emph{left-most} children from $\T_{k, h}^\ell$, as in the case without addition except we add fewer vertices in the worst-case.
If $\vec{t}_\ell < 0$, then we will instead subtract $|\vec{t}_\ell| \leq (k-1)/2$ \emph{right-most} children from $\T_{k, h}^\ell$.
In either case, the worst-case number of vertices added or subtracted is reduced to being at most $(k-1)/2$.

Now, we only consider the case of odd $k$ for a few reasons:
(1) if $k$ is odd, then we have symmetry in the digits, meaning we add at worst as many things as we may subtract.
This makes the analysis easier.
(2) it will turn out to be the case that even $k$ is not optimal, see \cref{appendix:even_k} for more details.
(3) if $k$ is odd, then there is a simple relation for when we use the root vertex of a tree or not to produce an output: we only use the root vertex if $t$ is in the right half of the tree.
This gives an intuitive condition on the height to support $T$ inputs: $h = \ceil{\log_k(2T)}$ -- set the height such that there is a sufficient number of left-half leaves.
This can be compared to standard $k$-ary trees, where the same condition is to use every leaf of the tree except the right-most one.

We are now ready to introduce our main contribution: \cref{alg:main_alg_simple}.
\begin{algorithm}[htb]
   \caption{$k$-ary Tree Mechanism with Subtraction}\label{alg:main_alg_simple}

   \begin{algorithmic}[1] %
    \STATE {\bfseries Input:} A stream $\vec{x}_1, \vec{x}_2, \dots, \vec{x}_T$, odd $k\geq 3$, privacy parameter $\epsilon$
    \STATE $h\gets \ceil{\log_k 2T}$
    \STATE Lazily draw $\vec{z}\sim \lap(h / \epsilon)^{(k^h-1)/2}$ 
    \FOR{$t=1$ {\bfseries to} $T$}
        \STATE $noise \gets 0$
        \STATE $\vec{t} \gets \nenc_k(t, h)$
        \STATE $p \gets 0$
        \FOR{$\ell=h$ {\bfseries to} $1$}
            \FOR{$j=1$ {\bfseries to} $|\vec{t}_\ell|$}
                \STATE $p \gets p + \sgn(\vec{t}_\ell)\cdot k^{\ell-1}$
                \STATE $noise \gets noise + \vec{z}_p$
            \ENDFOR
        \ENDFOR
        \STATE At time $t$, output $\sum_{j=1}^{t} \vec{x}_t + noise$
    \ENDFOR
\end{algorithmic}
\end{algorithm}
Structurally the algorithm releases prefix sums with a correlation in the noise that depends on the representation $\nenc_{k}(t, h)$ of each time step $t$.

\section{Proving \texorpdfstring{\cref{thm:main}}{main theorem}}
In this section, we will prove the properties of \cref{alg:main_alg_simple}, as summarized in \cref{thm:main}, starting with privacy.

\subsection{Privacy}
Instead of arguing about the privacy of~\cref{alg:main_alg_simple}, we will argue it for another algorithm with the same output distribution, \cref{alg:main_alg_full}, whose connection to trees is more immediate.
We proceed by establishing the equivalence of the output distributions in Algorithms~\ref{alg:main_alg_simple} and \ref{alg:main_alg_full}.
\begin{algorithm}[htb]
   \caption{More tree-like version of \cref{alg:main_alg_simple} for analysis}\label{alg:main_alg_full}

   \begin{algorithmic}[1] %
    \STATE {\bfseries Input:} $\vec{x}\in \{0, 1\}^T$, $k\geq 3$, privacy parameter $\epsilon$
    \STATE $h\gets \ceil{\log_k 2T}$
    \STATE Construct $k$-ary tree $\T_{k,h}$.
    \STATE For all $I\in\T_{k, h}$ compute $\nSigma(I) = \Sigma(I) + z$ for i.i.d.\ $z\sim\lap(h/\epsilon)$.
    \FOR{$t=1$ {\bfseries to} $T$}
        \STATE $count \gets 0$
        \STATE $\vec{t} \gets \nenc_k(t, h)$
        \STATE $p \gets 0$
        \FOR{$\ell=h$ {\bfseries to} $1$}
            \FOR{$j=1$ {\bfseries to} $|\vec{t}_\ell|$}
                \IF{$\sgn(\vec{t}_\ell) > 0$}
                    \STATE $count \gets count + \widehat{\Sigma}([p + 1, p + k^{\ell-1}])$
                    \STATE $p \gets p + k^{\ell-1}$
                \ELSE
                    \STATE $count \gets count - \widehat{\Sigma}([p - k^{\ell-1} - 1, p])$
                    \STATE $p \gets p - k^{\ell-1}$
                \ENDIF
            \ENDFOR
        \ENDFOR
        \STATE At time $t$, output $count$
    \ENDFOR
    \end{algorithmic}
\end{algorithm}
The following proposition will be helpful in establishing that the outputs are identically distributed.
\begin{proposition}\label{prop:valid_prefix_sums}
    When iterating over $t$, let $count_n$ and $p_n$ be the value of $count$ and $p$ after $1\leq n\leq\norm{\vec{t}}_1$ updates in \cref{alg:main_alg_full}.
    Then $count_n$ stores a noisy prefix sum up to $p_n\in [1, (k^h-1)/2]$.
\end{proposition}
\begin{proof}
    We show this by induction on $n$.
    Observe that $p_n\in [1, (k^h-1) / 2]$ for all valid $n$ since the first non-zero digit in $\vec{t}$ has to be positive, and since the greatest possible integer we can represent in these digits is $(k^h-1) / 2$.
    We therefore have that $p_1 \geq 1$ and $count_1 = \nSigma[1, p]$ is a noisy prefix sum up to $\vec{x}_p$.
    Assuming that after $n$ updates $count$ represents a noisy prefix sum up to $p$, we will argue that this still holds after the $n+1$\textsuperscript{st} update to $p$.
    Just before $count$ is updated, observe that $p=p_n\overset{\pmod{k^{\ell-1}}}{=} 0$, thus implying that the corresponding partial sums added or subtracted on lines 12 and 13 exist in the tree, and that the result will be a noisy prefix sum up to $p_{n+1}$.
\end{proof}
We can now prove that the two algorithms produce identically distributed outputs.
\begin{proposition}\label{prop:algs_equivalent}
    The output distribution of Algorithms~\ref{alg:main_alg_simple} and \cref{alg:main_alg_full} are equal.
\end{proposition}
\begin{proof}%
    First, observe that \cref{alg:main_alg_full} like \cref{alg:main_alg_simple} outputs the correct prefix sum on expectation which follows from \cref{prop:valid_prefix_sums} for $n=\norm{\vec{t}}_1$.
    We thus only have to argue for the equivalence of the noise added.
    With this target in mind, note that each time a vertex is used on lines~12 and 15 in \cref{alg:main_alg_full}, it can be uniquely identified by the $p$ after it is updated on the next line.
    The uniqueness can be inferred from the uniqueness of expressing the corresponding $\vec{p}=\nenc_k(p, h)$, where $p\in[1, (k^h-1)/2]$ can be interpreted as a time step.
    Following up, note that for any time step $t$, the sequence of $p$'s that are computed on line~10 for~\cref{alg:main_alg_simple} are the same as the sequence of $p$'s obtained after lines~13 and 16 for~\cref{alg:main_alg_full}.
    Since \cref{alg:main_alg_simple} explicitly indexes the noise terms by the same $p$, and since each term is identically distributed across the algorithm, it follows that each individual noise term that is added to the output at each step is identical up to their sign.
    
    What remains to account for is the lack of a factor \enquote{$\sgn(\vec{e}_\ell)$} on line~11 in~\cref{alg:main_alg_simple}.
    The reason why it is not needed comes from the fact that a given vertex in~\cref{alg:main_alg_full} always contributes with the same sign, that we always subtract right-most children and always add left-most children to produce prefix sum estimates.
    Since Laplace noise is symmetric, subtracting or adding noise has the same effect, and we can therefore drop the sign on line~11 in~\cref{alg:main_alg_simple} and still have the same output distribution, proving the statement.
\end{proof}
We are now ready to prove that \cref{alg:main_alg_simple} satisfies $\epsilon$-DP.
\begin{lemma}\label{lemma:main_private}
    \cref{alg:main_alg_simple} satisfies $\epsilon$-DP.
\end{lemma}
\begin{proof}%
    We proceed by proving the privacy of \cref{alg:main_alg_full}, the privacy of which implies the privacy of \cref{alg:main_alg_simple} via \cref{prop:algs_equivalent}.
    The privacy proof is standard and follows from identifying that if all $\nSigma(I)$ on line~4 are released privately, then the remaining computations are post-processing of that release.
    For neighboring inputs $\vec{x}'\sim \vec{x}$, let $i$ be the index at which they differ.
    Let $\vec{y}\in\R^{|\T_{k, h}|}$ be a vector where each entry is identified by a unique interval from $I \in\T_{k, h}$ equal to $\Sigma(I)$ given $\vec{x}$.
    Define $\vec{y}'$ analogously for $\vec{x}'$.
    Since any $i\in [1, T]$ is included in exactly $h$ intervals $I\in \T$, intentionally excluding the root vertex which is not used, it follows that the sensitivity for releasing $\vec{y}$ is $\Delta_1 = \norm{\vec{y} - \vec{y}'}_1 = h$.
    As $\vec{y}$ is released on line~4 using the Laplace mechanism with appropriate scale, and all computations afterward are post-processing of $y$, \cref{alg:main_alg_full} satisfies $\epsilon$-DP.
\end{proof}

\subsection{Utility}
As intuited, allowing for the subtraction of vertices does indeed improve utility.
Before stating the improvement, note that, analogous to the case without subtraction, the variance of the output from \cref{alg:main_alg_simple} at time $t$ will be equal to $(2h^2/\epsilon^2)\cdot\norm{\nenc_k(t, h)}_1$, since line~11 is executed $\norm{\nenc_k(t, h)}_1$ times.

\begin{lemma}\label{lemma:main_mse}
    \cref{alg:main_alg_simple} achieves a mean squared error of $\frac{k(1-1/k^2)h^3}{2\epsilon^2(1 - 1/k^h)}$ over $T=(k^h - 1)/2$ outputs.
\end{lemma}
\begin{proof}
    First note that the mean squared error over $T$ steps is equal to the average number of noise terms added together to produce an output, call this number $a_{[1, T]}$, times the variance of the noise in each vertex.
    We thus want to compute $a_{[1, T]} = \frac{1}{T}\sum_{t=1}^{T} \norm{\nenc_k(t, h)}_1$.
    Note that $T$ is the greatest positive integer that can be expressed using $h$ digits in $\{-\frac{k-1}{2},\dots, \frac{k-1}{2}\}$, and $-T$ the most negative integer.
    For any representation $\vec{t}=\nenc_{k}(t, h)$ where $t\in[1, T]$, note that we have a bijection where $-\vec{t}$ encodes a unique negative integer $\geq -T$.
    We therefore have that $a_{[1, T]} = a_{[-T, -1]}$, and $a_{[-T, T]} = (1-1/k^h)\cdot a_{[1, T]}$, where the extra factor comes from not including $0$.
    Since $a_{[-T, T]} = h\cdot\frac{1}{k}\sum_{d=-(k-1)/2}^{(k-1)/2} |d| = \frac{h(k^2-1)}{4k}$, by noting that every digit occurs the same amount of times in each position, we also have that $a_{[1, T]} = \frac{kh}{4}\frac{1-1/k^2}{1-1/k^h}$.
    Since $\var [\vec{z}_p] = 2h^2/\epsilon^2$, we arrive at a mean squared error of $\frac{kh}{4}\frac{1-1/k^2}{1-1/k^h}\cdot\frac{2h^2}{\epsilon^2} = \frac{kh^3(1-1/k^2)}{2\epsilon^2(1 - 1/k^h)}$.
\end{proof}
We follow up with a corollary that states how the error scales and what constant the optimal $k$ gives.
\begin{corollary}\label{cor:best_k_scaling}
    \cref{alg:main_alg_simple} achieves a mean squared error that scales as $\frac{k(1-1/k^2)}{2\epsilon^2\log(k)^3}\log(T)^3 + o(\log(T)^3)$, and which is minimized for $k=19$ yielding $(0.1236/\epsilon^2)\log(T)^3 + o(\log(T)^3)$.
\end{corollary}
In relation to Honaker Online, this is a constant improvement by more than a factor 4.

\subsection{Space and time complexity}

Finally, it turns out that extending to $k$-ary trees with subtraction does not impact the efficiency of tree-based aggregation.
We will need the following proposition to prove both the space and time complexity.
\begin{proposition}\label{prop:p_used_on_interval}
    Each noise term $\vec{z}_p$ used by $\cref{alg:main_alg_simple}$ will only be used for outputs in a time interval $[t_1, t_2]\subseteq[1, T]$.
\end{proposition}
\begin{proof}
    Let the value of the iterate $\ell$ be $q$ when $\vec{z}_p$ is added on line~11.
    Consider first the case when $q<h$.
    Note that for any time step $t$ with representation $\vec{t} = \nenc_k(t, h)$, $\vec{z}_p$ is only used if the $h-q$ most-significant digits are identical across $\vec{p}=\nenc_k(p, h)$ and $\vec{t}$.
    Interpreting any fixed number of most significant digits of $\nenc_k(t, h)$ as an integer, as $t$ is incremented, will produce a sequence of monotonically increasing integers, implying that any noise value $\vec{z}_p$ will start to get used at some $t_1$, and then after some time step $t_2\geq t_1$ will never be used again, proving the claim for $q < h$.
    If $q=h$, then once $\vec{z}_p$ is added for the first time at $t_1$, it will be added at every future time step including $t_2 = T$, proving the claim fully.
\end{proof}
\begin{lemma}\label{lemma:log_space}
    \cref{alg:main_alg_simple} runs in $O(k\log_k T)$ space.
\end{lemma}
\begin{proof}
    Note that at any time step $t$, we need to store noise values used for the current time step $t$, of which at most $h(k-1)/2$ are at most needed.
    Let $\vec{z}_p$ be an arbitrary noise value kept in memory at time $t$, and then no longer used starting from $t' \geq t$.
    By \cref{prop:p_used_on_interval}, $\vec{z}_p$ will never be used again after $t'$, and it can therefore be removed from memory.
    It follows that no noise values besides those used at a current time step need to be stored in memory, thus we need to store at most $h(k-1)/2 = O(k\log_k T)$ noise values.
\end{proof}

\cref{lemma:linear_time} states the time complexity of \cref{alg:main_alg_simple}, which assumes that (1) Laplace noise can be generated, and (2) $\nenc_k(t, h)$ can be computed, both in $O(k)$ time.
Note that the cost of (1) is unavoidable, and that (2) can be efficiently computed by incrementing the digits at each step, where an argument analogous to the analysis for incrementing a binary counter implies a time of $O(kT)$ for incrementing the counter $T$ times as measured in digit flips.
\begin{lemma}\label{lemma:linear_time}
    \cref{alg:main_alg_simple} can be implemented to run in $O(kT)$ time.
\end{lemma}
\begin{proof}
   First let $T=(k^h-1)/2$, meaning we support as many inputs as possible using $h$ digits.
   Note that there are exactly $T$ unique noise variables used by \cref{alg:main_alg_simple}, since $p$ will take on all values in $[1, T]$ when iterating over $t$.
   Rather than resetting $noise$ on each iteration of $t$, we can instead update it, removing noise terms that are no longer used and adding new noise as needed.
   Each of them will be added exactly once and subtracted at most once (\cref{prop:p_used_on_interval}), therefore we run in time $O(T)$ for this case.
   For $(k^{h-1}+1)/2 \leq T < (k^h-1)/2 = T'$, note that $\frac{T'}{T} \leq \frac{(k^h - 1)/2}{(k^{h-1} + 1)/2} \leq k$, and therefore if we estimate the runtime of $T$ by that for $T'$, we get a runtime of $O(kT')$ in this case.
\end{proof}
Finally, combining Lemmas~\ref{lemma:main_private},~\ref{lemma:main_mse},~\ref{lemma:log_space}~and~\ref{lemma:linear_time} gives us \cref{thm:main}.

\section{Laplace vs Gaussian noise}\label{sec:laplace_vs_gauss}

Having designed an efficient $\epsilon$-DP algorithm for continual counting with a state-of-the-art bound on the error, a natural question arises: how does it compare to \approxdp{} mechanisms?

Since the noise of the Gaussian mechanism grows as $\delta$ approaches zero, it is clear that for any given bounds on $\ell_1$ and $\ell_2$ sensitivity, using the Gaussian mechanism with sufficiently small $\delta$ will result in more noise than the Laplace mechanism.
We will see that for current continual counting methods this change happens for surprisingly large values of $\delta$, with the exact threshold depending on the number of time steps $T$.

Furthermore, we adapt the Laplace mechanism to work with $\ell_2$-sensitivity and show that it provides a meaningful \approxdp{} guarantee with variance close to that of the Gaussian mechanism.

\subsection{Laplace vs Gaussian noise for continual counting}\label{sec:laplace_ell1}
We start by stating a lemma that compares the currently best continual counting mechanisms under pure and approximate differential privacy.
\begin{lemma}\label{lemma:compare_errors}
    For constants $B_\epsilon, B_{\epsilon, \delta} > 0$, let $(B_\epsilon/\epsilon^2)\log(T)^3 + o\big(\log(T)^3\big)$ be a bound on the mean squared error of an $\epsilon$-DP mechanism, and $(B_{\epsilon, \delta}/\epsilon^2)\log(T)^2\log(1/\delta) + o\big(\log(T)^2\log(1/\delta)\big)$ be a bound on the mean squared error of an $(\epsilon, \delta)$-DP mechanism, both for continual counting. 
    Then the $\epsilon$-DP mechanism achieves equivalent or better mean squared error, up to lower-order asymptotically vanishing terms, for $\delta = O(1/T^{B_\epsilon / B_{\epsilon, \delta}})$.
\end{lemma}
\begin{proof}
    Comparing the leading terms of the errors and solving for $\delta$ yields the statement.
\end{proof}
\begin{proof}[Proof of \Cref{thm:beats_gaussian_noise}]
    \Cref{alg:main_alg_simple} has a leading constant of $B_\epsilon < 0.124$ for $k=19$ under \puredp{}, and the corresponding constant in \cite{henzinger_almost_2023} is $B_{\epsilon, \delta} = 4 / (\pi^2 \log(e)^3) > 0.1349$, so we have $B_\epsilon / B_{\epsilon, \delta} < 0.92$.
    \Cref{lemma:compare_errors} thus implies that the pure mechanism has a smaller bound on the noise for $\delta = O(T^{-0.92})$.
\end{proof}

\subsection{Laplace noise scaled to \texorpdfstring{$\ell_2$}{L2}-sensitivity}

We next consider how we may extend the Laplace mechanism with an \emph{approximate} DP guarantee similar to the Gaussian mechanism, i.e., with noise scaled to the $\ell_2$-sensitivity.
Since the proof of this is straightforward using known results we suspect this fact may be known, but we are not aware of it appearing in the literature.
The technique for extending the Laplace mechanism to approximate DP relies on the following heterogeneous composition theorem \cite{kairouz15}.
\begin{lemma}[Simplified version of Theorem 3.5 in \cite{kairouz15}]\label{lemma:composition_main}
    For any $\epsilon_j > 0$ for $j\in\{1,,\dots, k\}$, and $\delta\in (0, 1)$, the class of $\epsilon_j$-differentially private mechanisms satisfy $(\tilde{\epsilon}, \delta)$-differential privacy under $k$-fold adaptive composition for
    \begin{equation*}
       \tilde{\epsilon} = \min
       \bigg\{ 
           \sum_{j=1}^k \epsilon_j,\,
           \sum_{j=1}^k \frac{(e^{\epsilon_j} - 1)\epsilon_j}{e^{\epsilon_j} + 1} + \sqrt{2\ln(1/\delta)\sum_{j=1}^k \epsilon_j^2},
       \bigg\}  \enspace .
    \end{equation*}
\end{lemma}
Using \Cref{lemma:composition_main}, we can derive the following approximate DP guarantee for the $\ell_2$ Laplace mechanism.
\begin{lemma}[$\ell_2$ Laplace Mechanism]\label{lemma:laplace_approx}
    Let $f : \mathcal{X}^n \to \R^d$ be a function with $\ell_2$-sensitivity $\Delta_2 \coloneqq \max_{\D\sim\D'}\norm{f(\D)-f(\D')}_2$.
    For a given dataset $\D\in\mathcal{X}^n$ and $\epsilon, \delta \in (0, 1)$, the output $f(\D) + \lap(\Delta_2 / a_{\epsilon, \delta})^d$, where $a_{\epsilon, \delta} = \sqrt{2\ln(1/\delta)}(\sqrt{1 + \epsilon/\ln(1/\delta)}-1)$, satisfies \approxdp{}.
\end{lemma}
\begin{proof}
    For $x_1,\dots,x_j \in (0,1)$, since $\frac{e^{x_i} - 1}{e^{x_i} + 1} \leq x_i/2$ it follows from \Cref{lemma:composition_main} that $(\tilde{\epsilon},\delta)$-differential privacy is also satisfied for
    \begin{equation}\label{eq:simple-eps}
        \tilde{\epsilon} = \min
        \bigg\{ 
            \sum_{j=1}^k \epsilon_j,\,
            \frac{1}{2}\sum_{j=1}^k \epsilon_j^2 + \sqrt{2\ln(1/\delta)\sum_{j=1}^k \epsilon_j^2},
        \bigg\} \enspace .
    \end{equation}
    It suffices to establish the privacy of the Laplace mechanism when restricted to arbitrary neighboring inputs $x$, $x'$.
    We will reason about the privacy loss for each coordinate of $\vec{y} = f(x)$ separately relative to the neighboring output $\vec{y}' = f(x')$.
    Focusing on one coordinate $\vec{y}_j$ and its private release, we have that its release must satisfy $\epsilon_j$-differential privacy where $\epsilon_j = (a_{\epsilon, \delta}/\Delta_2)\cdot \abs{\vec{y}_j - \vec{y}_j'}$.
    Composing the $d$ private releases of each coordinate using the bound (\ref{eq:simple-eps}) and noting that $\sum_{j=1}^{d} \epsilon_j^2 \leq (a_{\epsilon, \delta}/ \Delta_2)^2 \sum_{j=1}^d (\vec{y}_j - \vec{y}_j')^2 \leq a_{\epsilon, \delta}^2$ therefore gives an $(\epsilon, \delta)$-guarantee where $\epsilon = a_{\epsilon, \delta}(\frac{a_{\epsilon, \delta}}{2} + \sqrt{2\ln(1/\delta)})$.
    Proceeding to solve for (positive) $a_{\epsilon, \delta}$ gives the statement.
\end{proof}
The variance needed for using Laplace noise instead of Gaussian noise for \approxdp{} is only a constant factor higher whenever $\epsilon \leq \log(1/\delta)$, and the constant approaches $2$ as $\epsilon$ approaches zero.

\begin{corollary}\label{cor:variance_laplace_vs_gaussian}
    Let $\vec{y}^{\mathrm{lap}}$ and $\vec{y}^{\mathrm{gauss}}$ be outputs based on the $\ell_2$ Laplace mechanism (\Cref{lemma:laplace_approx}) and the Gaussian mechanism (\Cref{lemma:gaussian_noise_scaling}) respectively, both to achieve the same \approxdp{} guarantee where $\epsilon \leq \ln(1/\delta)$.
    Then $\var[\vec{y}^{\mathrm{lap}}_j] \leq \frac{(\epsilon/\ln(1/\delta))^2}{2(\sqrt{1+\epsilon/\ln(1/\delta)} - 1)^2}\cdot\var[\vec{y}^{\mathrm{gauss}}_j]$. Also, $\lim_{\epsilon\to 0} \frac{(\epsilon/\ln(1/\delta))^2}{2(\sqrt{1+\epsilon/\ln(1/\delta)} - 1)^2} = 2$.
\end{corollary}
\begin{proof}
    Using \Cref{lemma:laplace_approx} and noting that the variance of the noise added is twice the square of the scale parameter for Laplace noise, we get that 
    \begin{equation*}
        \frac{\var[\vec{y}^{\mathrm{lap}}_j]}{\var[\vec{y}^{\mathrm{gauss}}_j]} = \frac{2\Delta_2^2/(2\ln(1/\delta)(\sqrt{1 + \epsilon/\ln(1/\delta)} - 1)^2)}{2\ln(1.25/\delta)\Delta^2_2 / \epsilon^2} \leq \frac{(\epsilon / \ln(1/\delta))^2}{2(\sqrt{1 + \epsilon/\ln(1/\delta)} - 1)^2} \enspace .
    \end{equation*}
    To show the limit, note that for positive $w \leq 1$ we have that $\frac{w}{\sqrt{1 + w} - 1} \leq \frac{w}{w/2 - w^2/8} = \frac{2}{1 - w/4} \to 2$ as $w\to 0$, and $\frac{w}{\sqrt{1 + w} - 1} \geq 2$, both based on truncated Taylor expansions of $\sqrt{1+w}$.
\end{proof}

Finally we prove \Cref{theorem:laplace_both} from the introduction.
\begin{proof}[Proof of \Cref{theorem:laplace_both}]
    From the proof of \Cref{lemma:laplace_approx}, we have that adding noise $\lap(\Delta_2/a)$ gives an $(\epsilon', \delta)$-guarantee where $\epsilon' = a(\frac{a}{2} + \sqrt{2\ln(1/\delta)})$ if $\epsilon',\delta\in(0, 1)$.
    Note that our setup is the same if we set $a = \frac{\Delta_2}{\lambda}$, and that adding noise $\lap(\lambda)$ gives an $\frac{\Delta_1}{\lambda}$-DP guarantee.
    We can therefore always guarantee $(\epsilon, \delta)$-DP where $\epsilon = \min\{\frac{\Delta_1}{\lambda}, \epsilon'\} < 1$ if we enforce $\lambda > \Delta_1$.
\end{proof}
Note that while functions with low $\ell_1$-sensitivity also have low $\ell_2$-sensitivity, the converse does not hold.
In principle we can replace Gaussian noise with Laplace noise at a limited loss in utility for approximate DP (\Cref{cor:variance_laplace_vs_gaussian}), but the corresponding \puredp{} guarantee afforded by the noise will be weak unless $\Delta_1$ is small.
\Cref{theorem:laplace_both} can alternatively be seen as adding an approximate DP guarantee to pure DP mechanisms built on top of the Laplace mechanism.

\section{Lower bound for continual observation on random inputs}\label{sec:lowerbound}

In this section we consider continual observation in the ``offline'' setting where a mechanism has access to the entire input and has to output a differentially private approximation of all prefix sums. 
This of course implies a lower bound for the ``online'' setting.
We show \Cref{thm:lower} which achieves the same lower bound that was previously known for worst-case inputs, in the setting where the input is a random string.

\begin{proof}[Proof of \Cref{thm:lower}]
    We will prove the theorem via a packing argument applied to a carefully constructed collection of random strings.
    Without loss of generality, let $B=\sqrt{T}$ be an integer divisible by $4$, $k\leq T^{1/5}$ be an even integer, and $T$ be large enough such that $k \leq B/4$.
    Also divide the timeline $[1, T]$ into blocks of length $B$: $B_1=[1,B], B_2=[B+1,2B], \dots, B_m=[T-B+1,T]$, $m=T/B$.
    For an arbitrary string $\vec{y}\in\{0, 1\}^T$, we use the notation $\vec{y}_{B_i}\in\{0,1\}^B$ for the substring of $\vec{y}$ formed from the indices in $B_i$.
    We use $\mathrm{err}(\cdot)$, taking as input the output of a mechanism, to denote the (random) maximum additive error over all the outputs.
    The argument is outlined below:
    \begin{itemize}
        \item Construct random strings $\vec{x}^{(0)},\dots, \vec{x}^{(m)}$ that are near-neighboring and have small total variation distance to random strings from $\mathcal{U}$.
        \item Argue that $\mathcal{M}$ assumes additive error $\leq\alpha$ on these inputs with constant probability by virtue of the small total variation distance.
        \item Show that for $i\ne i'$ $\mathcal{M}$ can be used to distinguish $\vec{x}^{(i)}$ from $\vec{x}^{(i')}$ unless $\alpha=\Omega(\varepsilon^{-1}\log T)$.
    \end{itemize}
    The intuition behind the argument is that an accurate mechanism $\mathcal{M}$ on $\mathcal{U}$ will also be accurate on our random (but highly correlated) strings $\vec{x}^{(0)},\dots \vec{x}^{(m)}$.
    Being accurate on these strings implies the ability to distinguish between them.
    However, since $\mathcal{M}$  satisfies $\varepsilon$-DP, distinguishing between the inputs with good confidence is impossible, implying a lower bound on the additive error $\alpha$. 

    \medskip
    
    To this end, we define the random variables $X^{(i)}_j = \norm{\vec{x}^{(i)}_{B_j}}_1 = \sum_{\ell\in B_j} \vec{x}^{(i)}_\ell$, i.e., as the sum of 1s in the block $B_j$ of $\vec{x}^{(i)}$.
    With the argument outlined, consider the random strings $\vec{x}^{(0)}, \vec{x}^{(1)},\dots, \vec{x}^{(m)}$ formed by the following process:
    \begin{enumerate}
        \item Draw $\vec{x}^{(0)}$ from $\mathcal{U}$ conditioned on $X^{(0)}_j\in [\frac{B}{4}, \frac{3B}{4}]$ for all $j\in [m]$. %
        Call this distribution~$\mathcal{D}^{(0)}$.
        \item For $i\in[m]$ derive $\vec{x}^{(i)}$ from $\vec{x}^{(0)}$ by choosing $k$ 0s in $B_i$ uniformly at random and flipping them to 1s.
        Call these distributions $\mathcal{D}^{(i)}$.
    \end{enumerate}
    
    We next define the algorithm $\Alg(\vec{y}, \vec{y}')$ outputting the difference between private sums on two arbitrary inputs $\vec{y}, \vec{y}'\in\{0, 1\}^T$.
    Formally, on receiving $\vec{y}$ and $\vec{y}'$, it runs $\M$ on each, producing outputs $\vec{a}_1, \dots, \vec{a}_T$ and $\vec{a}_1', \dots, \vec{a}_T'$ respectively, and at time $t$, the algorithm outputs $\Alg(\vec{y}, \vec{y}')_t = \vec{a}_t-\vec{a}_t'$.
    We will show, step-by-step, that the accuracy guarantee of $\M$ on $\mathcal{U}$ will imply distinguishing between our random strings $\vec{x}^{(0)}, \dots, \vec{x}^{(i)}$.
    
    Consider $\Alg(\vec{x}^{(i)}, \vec{x}^{(0)})$ for $i\in[1, m]$.
    Note that the true prefix sums of $\vec{x}^{(i)}$ and $\vec{x}^{(0)}$ at the end of each block will be identical up until block $B_i$, at the end of which they will differ by exactly $k$. 
    We leverage this observation by introducing the following disjoint events $E_j$.
    For $j\in[1, m]$ and arbitrary input strings $\vec{y}$, $\vec{y}'$, we define $E_j$ as the event that the following two statements are true: (1) for all $\ell\in[1, j-1] : \Alg(\vec{y}, \vec{y}')_{\ell B} \leq k/2$, and (2) $\Alg(\vec{y}, \vec{y}')_{jB} > k/2$.
    In other words, $E_j$ is the event that, when tracking only the outputs at the end of blocks, $\Alg(\vec{y}, \vec{y}')$ outputs a value greater than $k/2$ for the first time at the end of block $B_j$.
    Set $\alpha=k/4$.
    By our construction, if $\M(\vec{x}^{(0)})$ and $\M(\vec{x}^{(i)})$ each assumes additive error at most $\alpha$ when running $\Alg(\vec{x}^{(i)}, \vec{x}^{(0)})$, then $\Alg(\vec{x}^{(i)}, \vec{x}^{(0)})$ has additive error at most $2\alpha=k/2$, and $E_i$ happens with probability $1$.
    We thus have
    \begin{equation*}
        \Pr[\Alg(\vec{x}^{(i)}, \vec{x}^{(0)})\in E_i] \geq \Pr[\M(\vec{x}^{(i)}) \leq \alpha]\cdot\Pr[\M(\vec{x}^{(0)}) \leq \alpha]\enspace .
    \end{equation*}

    We want to argue that $\mathcal{M}$ is accurate on $\vec{x}^{(0)},\dots, \vec{x}^{(m)}$, with constant probability for each input taken separately, and we will argue via the total variation distance between each of these random strings, and the random string $\vec{x}\sim\mathcal{U}$.
    In general, for a (possibly randomized) algorithm $\mathcal{A} : S \to \R$, distributions $\mathcal{D}, \mathcal{D}'$ over $S$ with total variation distance $\tvd(\mathcal{D}, \mathcal{D}') \leq \Delta$, and an event $E$ over the output, we have that if $\Pr_{\mathcal{A}, x\sim\mathcal{D}}[\mathcal{A}(x)\in E] \geq p \geq 0$, then $\Pr_{\mathcal{A}, x\sim\mathcal{D}'}[\mathcal{A}(x)\in E] \geq p - \Delta$.
    This follows from the observation that, in the extreme case, $\Delta$ probability mass can be shifted from inputs $x\in S$ where $\Pr_{\mathcal{A}}[\mathcal{A}(x)\in E] = 1$ to inputs $x'\in S$ where $\Pr_{\mathcal{A}}[\mathcal{A}(x')\in E] = 0$.
    If we manage to show that, for $i\in[0, m]$, $\tvd(\vec{x}, \vec{x}^{(i)}) = o(1)$, then we can lower bound $\Pr[\Alg(\vec{x}^{(i)}, \vec{x}^{(0)})\in E_i]$ by a positive constant.
    We derive bounds on these total variation distances in the next two lemmas.
    \begin{lemma}\label{lemma:r0_uniformish}
       Let $\vec{x}\sim\mathcal{U}$ and $\vec{x}^{(0)}\sim\mathcal{D}^{(0)}$. Then $\tvd(\vec{x}, \vec{x}^{(0)}) \leq \Delta$ where $\Delta = 2\sqrt{T}\exp(-\sqrt{T}/8)$
    \end{lemma}
    \begin{proof}
     Let $S$ be the support of $\mathcal{U}$ and $S^{(0)} \subseteq S$ the support of $\mathcal{D}^{(0)}$. %
     We have that $X\sim \mathrm{Bin}(B, 1/2)$ describes the number of 1s in an arbitrary block $B_i$ of $\vec{x}$.
     Hoeffding's inequality gives $\Pr[\frac{B}{4} \leq X \leq \frac{3B}{4}] \geq 1 - 2\exp(-B/8)$ whereby it follows that $\Pr[\vec{x}\in S^{(0)}] \geq \big(1 - 2\exp(-B/8)\big)^{T/B}$.
     Observing that $\vec{x}^{(0)}$ is uniformly distributed over $S^{(0)}\subseteq S$, we proceed to directly compute the total variation distance: %
     \begin{align*}
     \tvd(\vec{x}, \vec{x}^{(0)}) &= \frac{1}{2}\sum_{\vec{s}\in S}|\Pr[\vec{x}=\vec{s}] - \Pr[\vec{x}^{(0)}=\vec{s}]|\\
     &= \frac{1}{2}|S^{(0)}|\bigg(\frac{1}{|S^{(0)}|}-\frac{1}{|S|}\bigg) + \frac{1}{2}(|S|-|S^{(0)}|)\cdot \frac{1}{|S|}
     = 1 - \frac{|S^{(0)}|}{|S|}
     = 1 - \Pr[\vec{x}\in S^{(0)}]\\
     &\leq 1 - \big(1-2\exp(-B/8)\big)^{T/B}\leq \frac{2T\exp(-B/8)}{B} = 2\sqrt{T}\exp(-\sqrt{T}/8)\enspace ,
     \end{align*} 
     where the last inequality is using Bernoulli's inequality and holds for $T$ larger than an absolute constant.
    \end{proof}
    We prove that the remaining random strings $\vec{x}^{(1)},\dots, \vec{x}^{(m)}$ are closely distributed to strings in $\mathcal{U}$ next.
    \begin{lemma}\label{lemma:ri_uniformish}
       Let $\vec{x}\sim \mathcal{U}$ and $\vec{x}^{(i)}\sim\mathcal{D}^{(i)}$ for $i\in[m]$.
       Then $\tvd(\vec{x}, \vec{x}^{(i)}) \leq \Delta'$ where $\Delta' = 10\cdot T^{-1/20}$.
    \end{lemma}
    \begin{proof}
        As the total variation distance satisfies the triangle inequality, we can write
        \begin{equation*}
            \tvd(\vec{x}, \vec{x}^{(i)}) \leq \tvd(\vec{x}, \vec{x}^{(0)}) + \tvd(\vec{x}^{(0)}, \vec{x}^{(i)})\enspace.
        \end{equation*}
        The first term we have already computed in \Cref{lemma:r0_uniformish}.
        To bound the second term, we will use a property of total variation distance: it cannot increase if one applies the same transformation to both arguments.
        Let $f : [m]\times [0, B] \to \mathcal{F}$ be a mapping to a family of distributions where $f(i, \ell)$ denotes the distribution described by a process where (1) a string $\vec{y}\sim\mathcal{D}^{(0)}$ is sampled and then (2) $\vec{y}_{B_i}$ is overwritten by a uniformly sampled string in $\{0, 1\}^B$, conditioned on it having $\ell$ 1s.
        For $i\in[m]$, we claim that ($f(i, X^{(i)}_i)$, $\mathcal{D}^{(i)})$ and $(f(i, X^{(0)}_i), \mathcal{D}^{(0)})$ are pairs of identical distributions.
        To see this for the first pair of distributions, consider $\vec{y}\sim f(i, X^{(i)}_i)$.
        By definition $\vec{y}_{B_j}$ is identically distributed to $\vec{x}^{(i)}_{B_j}$ for $j \neq i$, and so the only special case to consider is $j = i$.
        For this case, we note that $X^{(i)}_i$ is defined as the number of 1s in $\vec{x}^{(i)}_{B_i}$, and that symmetry implies that any substring with $\ell$ 1s is equally likely, thus $\vec{y}_{B_i}$ must be identically distributed to $\vec{x}^{(i)}_{B_i}$ as well, and we are done.
        An analogous argument can be used for showing that $f(i, X^{(0)}_i)$ and $\mathcal{D}^{(0)}$ are identically distributed.

        We thus have $\tvd(\vec{x}^{(0)}, \vec{x}^{(i)}) = \tvd(f(i, X^{(0)}_i), f(i, X^{(i)}_i)) \leq \tvd(X^{(0)}_i, X^{(i)}_i)$, and so our problem reduces to bounding this quantity.
        \begin{align*}
            \tvd(X^{(0)}_i, X^{(i)}_i) &= \frac{1}{2}\sum_{\ell=0}^{B}\lvert\Pr[X^{(0)}_i=\ell]-\Pr[X^{(i)}_i=\ell]\rvert\nonumber\\
            &= \frac{1}{2}\sum_{\ell=0}^{B}\lvert\Pr[X^{(0)}_i=\ell]-\Pr[X^{(0)}_i=\ell-k]\rvert\nonumber\\
            &= \frac{1}{2}\bigg(\sum_{\ell=B/4}^{B/4 + k -1}\Pr[X^{(0)}_i=\ell] + \sum_{l=3B/4+1-k}^{3B/4}\Pr[X^{(0)}_i=\ell]\nonumber\\
            &+ \sum_{\ell=B/4 + k}^{3B/4}\lvert\Pr[X^{(0)}_i=\ell]-\Pr[X^{(i)}_i=\ell]\rvert\bigg)\nonumber\\
            &\leq k\cdot\Pr[X^{(0)}_i=B/4+k] + \frac{1}{2}\sum_{\ell=B/4+k}^{3B/4}\lvert\Pr[X^{(0)}_i=\ell]-\Pr[X^{(0)}_i=\ell-k]\rvert
        \end{align*}
        The last step follows from the symmetry of $X^{(0)}_i$ : $\Pr[X^{(0)}_i = B/2 + \ell ] = \Pr[X^{(0)}_i = B/2 - \ell]$, and that the probability is greater for values closer to $B/2$.
        Observe that the first $B/4-k/2$ terms in the second sum of the last step are equal to the last $B/4-k/2$ terms by the symmetry of the distribution around $B/2$ (for $\ell=B/2+k/2$, the term in the sum is equal to zero). %
        Continuing the derivation and leveraging this symmetry, we get
        \begin{align*}
            \tvd(X^{(0)}_i, X^{(i)}_i) &\leq k\cdot\Pr[X^{(0)}_i=B/4+k] + \sum_{\ell=B/4+k}^{B/2+k/2-1}(\Pr[X^{(0)}_i=\ell]-\Pr[X^{(0)}_i =\ell-k])\\
            &= k\cdot\Pr[X^{(0)}_i=B/4+k]\\
            &+ \sum_{\ell=0}^{k-1}(\Pr[X^{(0)}_i=B/2-k/2+\ell]-\Pr[X^{(0)}_i=B/4+\ell])\\
            &\leq 2k\cdot \Pr[X^{(0)}_i=B/2]%
        \end{align*}
        where the first equality follows from the sum telescoping and the second inequality follows from bounding each term in the sum.
        Continuing, letting $X\sim\mathrm{Bin}(B, 1/2)$, we have that
        \begin{align*}
            2k\cdot\Pr[X^{(0)}_i = B/2] &= 2k\cdot\Pr\bigg[X = B/2 \,\bigg|\, \frac{B}{4} \leq X \leq \frac{3B}{4}\bigg]
            = \frac{2k\cdot\Pr[X = B/2]}{1-2\Pr[X < B/4]}\\
            &\leq \frac{2k\cdot\Pr[X = B/2]}{1-2\exp(-B/8)} \leq 6k\cdot\Pr[X=B/2] = 6k\cdot 2^{-B}\binom{B}{B/2}\\
            &\leq 6k\cdot2^{-B}\cdot\frac{\sqrt{2}\cdot 2^{B}}{\sqrt{\pi B}} = \frac{6\sqrt{2}}{\sqrt{\pi}}\cdot\frac{k}{\sqrt{B}} = \frac{6\sqrt{2}}{\sqrt{\pi}} \cdot T^{-1/20}\enspace ,
        \end{align*}
        where all inequalities are valid for $T$ greater than some absolute constant, and the last inequality is a variant of Stirling's inequality.
        We finally get
        \begin{equation*}
        \tvd(\vec{x}, \vec{x}^{(i)}) \leq \tvd(\vec{x}, \vec{x}^{(0)}) + \tvd(\vec{x}^{(0)}, \vec{x}^{(i)}) \leq 2\sqrt{T}\exp(-\sqrt{T}/8) + \frac{6\sqrt{2}}{\sqrt{\pi}} \cdot T^{-1/20} \leq 10\cdot T^{-1/20}\enspace ,
        \end{equation*}
        where the last inequality is true for $T \geq 400$, thus proving the lemma. %
    \end{proof}
    Using \Cref{lemma:r0_uniformish} and \Cref{lemma:ri_uniformish}, we arrive at
    \begin{align*}
        \Pr[\Alg(\vec{x}^{(i)}, \vec{x}^{(0)})\in E_i] &\geq \Pr[\mathrm{err}(\mathcal{M}(\vec{x}^{(i)})) \leq \alpha]\cdot\Pr[\mathrm{err}(\mathcal{M}(\vec{x}^{(0)}))\leq \alpha]\\
        &\geq \big(\Pr[\mathrm{err}(\mathcal{M}(\vec{x}))\leq \alpha]-\tvd(\vec{x}^{(i)}, \vec{x})\big)\\
        &\times\big(\Pr[\mathrm{err}(\mathcal{M}(\vec{x}))\leq \alpha]-\tvd(\vec{x}^{(0)}, \vec{x})\big)\\
        &\geq (3/4 - \Delta')(3/4 -\Delta) \geq 1/2\enspace .
    \end{align*}
    where the last inequality holds for $T$ greater than an absolute constant.
    Next observe that $\vec{x}^{(i)}$ and $\vec{x}^{(0)}$ are $k$-neighboring with probability 1, and so we have $\Pr[\Alg(\vec{x}^{(i)}, \vec{x}^{(0)})\in E_i] \leq e^{k\varepsilon}\Pr[\Alg(\vec{x}^{(0)}, \vec{x}^{(0)})\in E_i]$, and therefore $\Pr[\Alg(\vec{x}^{(0)}, \vec{x}^{(0)})\in E_i]\geq e^{-k\varepsilon}/2$. 
    Since all $E_j$ are disjoint, we get 
    \begin{align*}
        1\geq \sum_{j=1}^{m}\Pr[\Alg(\vec{x}^{(0)}, \vec{x}^{(0)})\in E_j]\geq m\cdot e^{-k\varepsilon}/2 \enspace,
    \end{align*}
    and therefore
    \begin{align*}
        k\geq \varepsilon^{-1}\log(m/2)\enspace .
    \end{align*}
    Recalling that $m=T/B = \sqrt{T}$, we thus have $k=\Omega(\varepsilon^{-1}\log T)$ and therefore $\alpha= k/4 =\Omega(\varepsilon^{-1}\log T)$.
\end{proof}

\section{Discussion}
While there are no known fundamental reasons ruling out asymptotic improvements in the mean squared error for \puredp{} -- it might be possible to match the $\Omega(\log(T)^2)$ lower bound of~\cite{henzinger_almost_2023} -- there have been no asymptotic improvements since the binary tree mechanism \cite{dwork_differential_2010, chan_differentially_2012}.
This state of affairs motivates exploring how much leading constants can be improved using the factorization mechanism framework.
\cref{alg:main_alg_simple} falls into the framework, is easy to implement, efficient, and offers the best known error-scaling to date.

Our paper also reveals that using Laplace noise improves the best possible $(\varepsilon,\delta)$-differentially private factorization mechanisms based on Gaussian noise whenever $\delta$ is sufficiently small.
An interesting future direction is to investigate, for fixed $\epsilon, \delta$ and a given factorization: which noise distribution achieves optimal error?
It seems unlikely that Laplace or Gaussian noise is always best -- indeed, Vinterbo~\cite{Vinterbo22} has results of this flavor in more restricted settings.

Finally, while we show a lower bound for random inputs that matches the best lower bound known for worst-case inputs, it is still possible that random inputs are easier and allow lower error than what can be achieved for worst-case inputs.

\medskip

{\bf Acknowledgement.} We are grateful to Monika Henzinger and Jalaj Upadhyay for discussions on continual counting that inspired part of this work.

\bibliographystyle{plainurl}%
\bibliography{main}

\newpage
\appendix

\section{Definitions for Differential Privacy}\label{appendix:dp}

Here are the standard definitions for differential privacy which we have use for in the paper.
\begin{definition}[$(\epsilon, \delta)$-Differential Privacy~\cite{dwork_algorithmic_2013}]
    A randomized algorithm $\M : \mathcal{X}^n\to\mathcal{Y}$ is ($\epsilon, \delta)$-differentially private if for all $S\subseteq\mathsf{Range}(\M)$ and all neighboring inputs $\mathcal{D}, \mathcal{D}'\in \mathcal{X}^n$, the neighboring relation written as $\mathcal{D}\sim \mathcal{D}'$, we have that:
    \begin{equation*}
        \Pr[\M(\mathcal{D})\in S] \leq \exp(\epsilon)\Pr[\M(\mathcal{D}')\in S] + \delta\, ,
    \end{equation*}
    where $(\epsilon, 0)$-DP is referred to as \puredp{}.
\end{definition}
\begin{lemma}[Laplacian Mechanism \cite{dwork_algorithmic_2013}]
    Let $f : \mathcal{X}^n \to \R^d$ be a function with $\ell_1$-sensitivity $\Delta_1 \coloneqq \max_{\D\sim\D'}\norm{f(\D)-f(\D')}_1$.
    For a given data set $\D\in\mathcal{X}^n$ and $\epsilon > 0$, the mechanism that releases $f(\D) + \lap(\Delta_1/\epsilon)^d$ satisfies \puredp{}.
\end{lemma}
\begin{lemma}[Gaussian Mechanism \cite{dwork_algorithmic_2013}]\label{lemma:gaussian_noise_scaling}
    Let $f : \mathcal{X}^n \to \R^d$ be a function with $\ell_2$-sensitivity $\Delta_2 \coloneqq \max_{\D\sim\D'}\norm{f(\D)-f(\D')}_2$.
    For a given data set $\D\in\mathcal{X}^n$ and $\varepsilon, \delta \in (0, 1)$, the mechanism that releases $f(\D) + \mathcal{N}(0, 2\ln(1.25/\delta)\Delta_2^2/\epsilon^2)^d$ satisfies \approxdp.
\end{lemma}
\section{\texorpdfstring{$k$}{k}-ary Tree Mechanisms: Proofs}
In this subsection, we will briefly discuss the proof sketch for Theorem \ref{thm:no_subtraction}. It is important to note that while this proof is included, it is not considered a main contribution of this paper.
\begin{proof}[Proof sketch of Theorem~\ref{thm:no_subtraction}] \label{appendix:proofT1}
    The utility and privacy follow from Lemmas \ref{lemma:nosub_private}~and~\ref{lemma:nosub_mse}, and the online nature of the mechanism follows from the partial sums being added on Line~11 involving inputs $\vec{x}_i$ with $i\leq p\leq t$.
    For the argument about space usage, a standard argument \cite{chan_private_2011, andersson2024smooth} can be extended from the binary tree mechanism to $k$-ary tree mechanism
    to yield $O(k\log_k T)$.
    Similarly, the time needed to release $T$ prefix sum will be $O(kT)$ (assuming noise can be generated in constant time) based on extending the analysis of incrementing a binary counter to incrementing a $k$-ary counter.
\end{proof}
\section{Leveraging subtraction for even \texorpdfstring{$k$}{k}}\label{appendix:even_k}

In the main part of the paper we only considered subtraction for $k$-ary trees when $k$ was odd.
This part of the appendix is intended for motivating that choice.

First note that for even $k$ we lose the symmetry in the offset digits, so we need to make a decision if the digits in the offset representation are $\{-k/2 + 1, \dots k/2\}$ or $\{-k/2, \dots, k/2 - 1\}$.
We pick the first one as that allows for supporting more inputs, giving \cref{def:even_basis_offset}.
\begin{definition}[Offset (even) $k$-ary representation of integers]\label{def:even_basis_offset}
    Let $k$, $w$ and $t$ be positive integers, with $k$ also being even, satisfying $t \leq \frac{k(k^h - 1)}{2(k-1)}$.
    We then define $\vec{t} = \eenc_k(t, w)\in[-\frac{k}{2}+1, \frac{k}{2}]^w$ as the unique vector with integer entries in $[-\frac{k}{2}+1, \frac{k}{2}]^w$ satisfying $\sum_{i=1}^{w}k^{i-1}\vec{t}_i = t$.
\end{definition}
Substituting the representation used in \cref{alg:main_alg_simple} by the one in \cref{def:even_basis_offset} and setting the height appropriately gives a valid algorithm using a $k$-ary tree.
However, it performs worse than using odd $k$.
\begin{claim}\label{claim:even_mse}
    For even $k\geq 4$, $h\geq 1$, and $T = \frac{k(k^h - 1)}{2(k-1)}$, \cref{alg:main_alg_simple}, with $h$ correctly set, achieves a mean squared error of $\frac{kh^3}{2\epsilon^2}\frac{1}{1-1/k^h} + o(h^2)$ when producing $T$ outputs.
\end{claim}
\begin{proof}[Proof sketch.]
    The same proof as for \cref{lemma:main_mse} does not work due to the asymmetry in the digits.
    Instead, we can express a recursion for the number of terms needed.
    Let $c_h$ be the total number of vertices combined to produce all prefixes that a tree of height $h$ can support.
    Then we have
    \begin{equation*}
        c_h = \bigg(\frac{k+2}{4} + \frac{k(h-1)}{4}\bigg)\frac{k}{2}k^{h-1} + c_{h-1}\,,
    \end{equation*}
    where $c_0 = 0$.
    The idea behind the recursion is as follows: the first term corresponds to all the integer representations with positive leading digit, of which there are $k/2 \cdot k^h$, and the leading digit contributes $(k+2)/4$, the remaining $k(h-1)/4$ on average.
    The second term is the case where the first digit is zero, in which case the contribution is $c_{h-1}$.
    Solving the recursion and dividing by $T$ yields the average number of vertices added, $c_h / T = \frac{kh}{4}\frac{1}{1-1/k^h} + o(1)$, which when multiplied by the variance of a vertex, $2h^2/\epsilon^2$, gives the statement.
\end{proof}
Ignoring the lower order terms, this is worse by a factor $1- 1/k^2$ over the odd case, and we can also find the optimal $k$ here.
\begin{corollary}\label{cor:best_k_scaling_even}
    \cref{alg:main_alg_simple}, with modifications for even $k$, achieves a mean squared error that scales as $\frac{k}{2\epsilon^2\log(k)^3}\log(T)^3 + o(\log(T)^3)$, and which is minimized for $k=20$ yielding $(0.1238/\epsilon^2)\log(T)^3 + o(\log(T)^3)$.
\end{corollary}
As the analysis is less elegant, we get more lower order terms and the scaling is worse, we focus on odd $k$ in the paper.

\end{document}